\documentclass[publicationstatus]{eptcs}
\providecommand{\event}{SR 2015} 

\usepackage[fleqn]{amsmath}
\usepackage{amssymb,latexsym,wasysym,dsfont,stmaryrd}
\usepackage{relsize}
\usepackage{xspace}
\usepackage{graphicx}
\usepackage{gastex}
\usepackage{wrapfig}
\usepackage{mymacros,macrosadd,sectionningmacros}
\usepackage{manfnt}              
\usepackage[dvipsnames]{xcolor}
\usepackage{enumitem}
\setlist[enumerate]{itemsep=0mm}

\newif\ifdraft\draftfalse

\ifdraft
\newcommand\modedraft[1]{#1}
\newcommand\todo[1]{{\color{purple}[\textbf{To do:} #1]}}
\newcommand\bmcomment[1]{\marginpar[{\color{blue}\small\dbend}]{\color{blue}\small\dbend}{\footnotesize \color{blue}[#1 - \textbf{Bastien}]}}
\newcommand\skcomment[1]{\marginpar[{\color{purple}\small\dbend}]{\color{OliveGreen}\small\dbend}{\footnotesize \color{purple}[#1 - \textbf{Sophia}]}}
\newcommand\bm[1]{{\color{blue}{#1}}}
\newcommand\sk[1]{{\color{OliveGreen}{#1}}}
\else
\newcommand\modedraft[1]{}
\newcommand\todo[1]{}
\newcommand\bm[1]{}
\newcommand\sk[1]{}
\newcommand\bmcomment[1]{}
\newcommand\skcomment[1]{}
\fi

\title{Dealing with imperfect information in Strategy Logic}
\author{Sophia Knight
\institute{LORIA - CNRS / Universit\'e de Lorraine\\ Nancy, France}
\email{sophia.knight@gmail.com}
\and
Bastien Maubert
\institute{LORIA - CNRS / Universit\'e de Lorraine\\
Nancy, France}
\email{bastien.maubert@gmail.com}
}
\def\titlerunning{Dealing with imperfect information in Strategy Logic}
\def\authorrunning{S. Knight \&  B. Maubert}
\begin{document}
\maketitle

\begin{abstract}
We propose  an extension of Strategy Logic (SL), in which one can both reason about
strategizing under imperfect information and about players' knowledge. One
original aspect of our approach is that we do not force strategies to be
uniform, i.e. consistent with the players' information, at the
semantic level; instead, one can express in the logic itself that a
strategy should be uniform. To do so, we first develop a
``branching-time'' version of SL  with perfect information, that we
call BSL, in which
one can quantify over the different outcomes defined by a partial assignment
of strategies to the players; this contrasts with SL, where temporal operators are
allowed only when all strategies are fixed, leaving only one possible
play.
Next, we further extend BSL by adding distributed knowledge operators, the semantics of which rely on equivalence relations on partial
plays. The logic we obtain subsumes most strategic logics with
imperfect information, epistemic or not.
\end{abstract}


\section{Introduction}
\label{sec-intro}

Over the past decade, investigation of logical systems 
for studying strategic abilities has thrived in the areas of artificial
intelligence and multi-agent systems. However, there is still no
 satisfying logical
framework  to model, specify and analyze such systems.
One of the  proposals most studied so far is  Alternating-time Temporal Logic (\ATL)
\cite{DBLP:journals/jacm/AlurHK02}, 
in which one can
specify  what objectives coalitions of agents can achieve. 
Several extensions were introduced ($\ATLs$, game
logics\ldots), but
all of these logics fail to model non-cooperative situations where agents
follow individual objectives. 
It is
well known that
studying  this kind of situation requires solution concepts from game theory,
such as Nash equilibria, that cannot be expressed in
\ATL or its extensions.

To address this shortcoming, Chatterjee, Henzinger and
Piterman recently introduced Strategy
Logic (\SL) \cite{DBLP:journals/iandc/ChatterjeeHP10}. This logic subsumes all
extensions of \ATL, and because it considers strategies as first-order
citizens in the language, it can express
fundamental game-theoretic concepts such as Nash Equilibria or
dominated strategies. 
\SL has recently been extended
and intensively studied 
 \cite{DBLP:conf/fsttcs/MogaveroMV10,DBLP:conf/concur/MogaveroMPV12,DBLP:journals/tocl/MogaveroMPV14}. 
Relevant fragments enjoying nice
computational characteristics have been identified. In particular, the
syntactic fragment \SLOG (One-Goal Strategy Logic) is strictly more
expressive than \ATLs, but
not computationally more expensive \cite{DBLP:conf/concur/MogaveroMPV12}. 

However, despite its great expressiveness, 
there is one fundamental feature of
most real-life situations
that SL lacks, which is imperfect information. 
An agent 
has imperfect information  if
she does not know 
 the exact  state of the system at every moment, but only has
 access to an
 approximation of it. 
Considering  agents with imperfect information raises two major
theoretical issues. The first one concerns strategizing under
imperfect information. Indeed, in this context an agent's strategy must prescribe  the same
choice in all situations that are  indistinguishable to the
agent. Such strategies are called \emph{uniform strategies}, and this requirement
deeply impacts  the task of computing strategies \cite{reif84}.
The second
main theoretical challenge relates to uncertainty, deeply intertwined
with imperfect information, and it consists of representing and reasoning about agents'
knowledge. Over the past decades, much effort has been put  into
devising logical systems that address
this issue, first in static settings \cite{fagin1995reasoning} and later
 adding
 dynamics \cite{halpern2004complete,van2011logical}. 

Concerning \ATL, many variants have been introduced that deal with
imperfect information
\cite{van2003cooperation,DBLP:journals/fuin/JamrogaH04,schobbens2004alternating,jamroga2006agents}. Some
of these
numerous logics
 deal
with strategizing under imperfect information, some with reasoning about
knowledge; because it is not natural to reason about the knowledge of
agents with imperfect information without treating the strategic aspects
accordingly, as argued in \cite{DBLP:journals/fuin/JamrogaH04}, some
treat both aspects. But there still remain a number of 
logics that do so, and that essentially differ in the semantics of the strategic operator:
how much memory do agents have?  should  agents simply have a strategy to achieve
some goal? Or should they
\emph{know that there is} a strategy to do so? Or \emph{know  a
  strategy} that works? The two last notions are usually referred to as \emph{de
  dicto} and \emph{de re} strategies, respectively \cite{DBLP:journals/fuin/JamrogaH04}.

About \SL, very few works have considered imperfect information. 
\cite{DBLP:journals/corr/Belardinelli14} and \cite{DBLP:conf/cav/CermakLMM14}
propose epistemic extensions of \SL, but they do not require
strategies to be uniform  \ie
being consistent with the agents' information. In \cite{belardinelli2015}, an
epistemic strategy logic is proposed in which uniform strategies are
considered, and interestingly the \emph{de re} semantics of the
strategic operator can be
expressed in the \emph{de dicto} semantics, providing some
flexibility. However, how much memory strategies use \bmcomment{we
  hardwire it a little in the path relation}, and whether
they should be uniform or not, still has to be hardwired in the
semantics.

In this work, we propose yet another epistemic strategy logic, with
the purpose of getting rid of the constraint of enforcing what kind of
strategies are to be used at the semantic level. To do so, we first
develop a ``branching-time'' version of SL with perfect information.
In \SL,  temporal operators are
allowed only when all strategies are fixed, leaving only one possible
play. We relax this constraint 
by introducing a \emph{path quantifier}, which quantifies over the
different outcomes defined by a partial assignment of strategies to
the agents.  This enables the comparison of the various outcomes of a
strategy. Because it will be important, for instance to express the
uniformity of a strategy, to consider all the possible outcomes of a
strategy  assigned to an agent $a$, we need a way to remove in
an assignment the bindings of all agents but $a$. We thus introduce an
\emph{unbinding operator}. We call the resulting logic Branching-time
Strategy Logic (\BSL), and we prove by providing linear translations
in both directions that it has the same
expressive power and same computational complexity as \SL.  We also
present a variant of \BSL, called \BSLplus, which can in addition
refer to the actions chosen by each agent at each moment, and we
conjecture that it is strictly more expressive than \SL and \BSL.
Next, we define our Epistemic Strategy Logic (\ESL) by further
extending \BSL with distributed knowledge operators, the semantics of
which rely on equivalence relations on partial plays. We do not change the
semantics of the strategy quantifier to require them to be uniform, or
\emph{de re}, or \emph{de dicto}, or \emph{memoryless}, but we rather
show that all of these properties of strategies can be expressed in
the language, which thus subsumes most, if not all, the variants of
epistemic strategic logics with imperfect information that we know about.

The paper is structured as follows. In Section~\ref{sec-prelim} we recall the
models, syntax and semantics of \SL. In Section~\ref{sec-branching},
we define \BSL and \BSLplus, and we prove that \SL and \BSL are
equiexpressive. We then present \ESL in Section~\ref{sec-epistemic},
where we also
 show how it can express various classic properties of
strategies. We conclude and discuss future work in Section~\ref{sec-conclusion}.
Some proofs are omitted by lack of space.



\section{Preliminaries}
\label{sec-prelim}

Let $\AP$ be a countable non-empty set of \emph{atomic propositions}, 
$\Ag$ a non-empty finite set of \emph{agents} and $\Act$ a non-empty finite set of
\emph{actions}. We let $\AcPr=\Act^{\Ag}$ be the set of possible
\emph{decisions}. For $\acpr\in\AcPr$ and $a\in\Ag$, $\acpr(a)$
is the action taken by Agent~$a$ in decision $\acpr$.

\subsection{Concurrent game structures}

A \emph{concurrent game structure} (CGS) is a tuple
$\ga=(\sstates,\delta,\state_\init,\gval)$, where $\sstates$ is a
countable non-empty set of \emph{states}, $\delta:\sstates\times \AcPr
\to \sstates$ is a \emph{transition function}, $\state_\init$ is the
\emph{initial state} and $\gval:\sstates\to 2^{\AP}$ is a
\emph{valuation function}. A \emph{path} is an infinite word
$\path=\state_0 (\acpr_1, \state_1) \ldots \in \sstates\cdot
(\AcPr\times \sstates)^\omega$ such that for all $i\geq 0$,
$\state_{i+1}=\delta(\state_i,\acpr_{i+1})$, and an \emph{initial
  path} $\fpath$ is a finite prefix of a path. In the following, we
shall write $\state_0\acpr_1 \state_1\ldots$ instead of
$\state_0(\acpr_1,\state_1)\ldots$, and similarly for initial paths.
For a state $\state$ we denote by $\PathsInf(\state)$
(resp. $\PathsFin(\state)$) the set of paths (resp. initial paths)
that start in $\state$, \ie for which $\state_0=\state$. We also let
$\PathsInf$ (resp. $\PathsFin$) be the set of all paths (resp. initial
paths).  For a path $\path=\state_0\acpr_1\state_1\ldots$, for
$i,j\geq 0$, we let $\path[i]\egdef \state_i$, $\path_{\leq i}\egdef
\state_0\ldots \acpr_i\state_i$, $\path_{\geq i}\egdef
\state_i\acpr_{i+1}\state_{i+1}\ldots$ and $\path[i,j]\egdef \state_i\acpr_{i+1}\ldots\acpr_{j}\state_j$. For an initial path
$\fpath=\state_0\ldots \acpr_n\state_n$, $\last(\fpath)\egdef
\state_n$ is its last state and $|\fpath|\egdef n$ is the index of its
last state.  Given two initial paths $\fpath=\state_0\acpr_1
\state_1\ldots\acpr_n\state_n$ and $\fpath'=\state'_0\acpr'_1
\state'_1\ldots\acpr'_m\state'_m$ such that $\state_n=\state'_0$, we
let $\fpath\cdot\fpath'\egdef\state_0\acpr_1
\state_1\ldots\acpr_n\state_n \acpr'_1
\state'_1\ldots\acpr'_m\state'_m$ be their concatenation.

A \emph{strategy} is a total function
$\strat:\PathsFin\to\Act$ that assigns an action to each initial path, and we
let $\setstrat$ be the set of all strategies. 
Also, given a 
strategy $\strat$ and an initial path
$\fpath\in\PathsFin$ ending in state $\state$,
 we define the
\emph{$\fpath$-translation} of $\strat$ as the 
strategy $\strattrans{\strat}$ such that
for all
initial paths $\fpath'\in\PathsFin(\state)$,
$\strattrans{\strat}(\fpath')\egdef\strat(\fpath\cdot\fpath')$, and 
for all
initial paths $\fpath'\in\PathsFin(\state')$ with $\state'\neq\state$, $\strattrans{\strat}(\fpath')=\strat(\fpath')$.


Let $\Var$ be a countably infinite set of \emph{variables}. An \emph{assignment} is
a partial function $\assign:\Ag\union\Var \partialto \setstrat$, assigning to
each  agent and variable in its domain a strategy. 
For an assignment $\assign$, an agent $a$ and a strategy $\strat$,
$\assign[a\mapsto\strat]$ is the assignment of domain
$\dom(\assign)\union\{a\}$ that maps $a$ to $\strat$ and is equal to
$\assign$ on the rest of its domain, and similarly for
$\assign[\var\mapsto \strat]$ where $\var$ is a variable; also,
$\assign[a\mapsto\unb]$ is the assignment of domain
$\dom(\assign)\setminus\{a\}$, on which it is equal to $\assign$.
Given an assignment $\assign$ and a state $\state$, we define the
\emph{outcome} of $\assign$ in $\state$, written
$\out(\state,\assign)$, as the set of paths $\pi=\state_0\acpr_1
\state_1\ldots$ such that $\state_0=\state$, and for all $k\geq 0$,
for every agent $a$ in the domain of $\assign$,
$\acpr_{k+1}(a)=\assign(a)(\pi_{\leq k})$. We say that an assignment
$\assign$ is \emph{complete} if it assigns a strategy to each agent,
\ie $\Ag\subseteq\dom(\assign)$.  Given an assignment $\assign$ and an
initial path $\fpath$ ending in state $\state$, we define the
$\fpath$-translation of $\assign$ as the assignment
$\assigntrans{\assign}$ such that
$\dom(\assigntrans{\assign})=\dom(\assign)$, and for all
$l\in\dom(\assigntrans{\assign})$,
$\assigntrans{\assign}(l)\egdef\strattrans{\assign(l)}$ ($l$ being either a variable or an agent). 

Finally, we want (some of) our logics to be able to talk about the
precise actions taken by agents. To do so, we consider the following set of
\emph{action propositions}: $\APact\egdef\{p^a_\act \mid \act\in \Act \mbox{
  and }a\in\Ag\}$, and we let $\APplus\egdef\AP\uplus\APact$. In the
following, we will therefore always assume that CGSs are
\emph{unfolded}, such that each state $\state$ is reached by one unique
transition through some decision $\acpr_\state$, except the initial state $\state_\init$
which has no incoming transition. We can thus extend the valuation function
$\gval$ into $\gvalp$ as
follows: $\gvalp(\state_\init)\egdef\gval(\state_\init)$, and for
every state $\state\neq\state_\init$,
$\gvalp(\state)\egdef\gval(\state)\union \{p^a_{\acpr_\state(a)}\mid a\in\Ag\}$.

\subsection{Strategy Logic}
\label{sec-SL}


We recall the syntax and semantics of Strategy Logic (\SL). 
First, the set of formulas in \SL is given by the following grammar:
\[\phi \egdef p \mid \neg \phi \mid \phi\ou\phi \mid \X \phi \mid \phi
\until \phi \mid \Estrat\phi \mid (a,\var)\phi\]
where $p\in\AP$, $\var\in\setvar$ and $a\in\Ag$. 

Notice that \SL-formulas cannot talk about agents' actions. 

We define $\top$ as $p\ou\neg p$. Dual operators can be defined as usual:
$\perp\egdef\neg\top,\phi\et\phi'\egdef\neg(\neg\phi\ou\neg\phi'),\phi\release\phi'\egdef\neg(\neg\phi\until\neg\phi')$
and $\Astrat\phi\egdef\neg\Estrat\neg\phi$, and we also define the
classic temporal operators ``eventually'' and ``always'':
$\F\phi\egdef\top\until\phi$, and $\G\phi\egdef\phi\until\perp$.
Recall that $\Estrat$ is the \emph{strategy quantifier}, and $(a,\var)$
is the \emph{binding operator}: $\Estrat\phi$ reads as ``there exists
a strategy $\var$ such that $\phi$'', and $(a,\var)\phi$ reads as ``$\phi$
holds after agent $a$ is bound to the strategy denoted by $\var$''. 

For a formula $\phi$, $\free(\phi)\subseteq\Var$ is the set of
\emph{free variables} in $\phi$, \ie the set of variables $\var$ that occur
in $\phi$ without being under the scope of some quantification $\Estrat$.
In the following, given a formula $\phi$, an \emph{assignment for
  $\phi$} refers to an assignment $\assign$ such that $\free(\phi)\subseteq\dom(\assign)$.

Let $\phi$ be an \SL-formula. Given a CGS $\ga=(\sstates,\delta,\state_\init,\gval)$, an 
assignment
$\assign$ for $\phi$ and a state
$\state\in\sstates$, the semantics of $\phi$ in $\ga$ with assignment
$\assign$ at state $\state$ is defined inductively as follows:

\begin{tabular}{lcl}
 $\ga,\assign,\state\modelsSL p$ & if & $p\in\gval(\state)$\\
 $\ga,\assign,\state\modelsSL \neg\phi$ & if &
  $\ga,\assign,\state\not\modelsSL\phi$\\
 $\ga,\assign,\state\modelsSL \phi\ou\phi'$ & if &
  $\ga,\assign,\state\modelsSL\phi$ or
  $\ga,\assign,\state\modelsSL\phi'$ \\
 $\ga,\assign,\state\modelsSL\Estrat\phi$  & if & there exists
  $\strat\in\setstrats$ such that 
  $\ga,\assign[\var\mapsto\strat],\state\modelsSL \phi$\\
 $\ga,\assign,\state\modelsSL(a,\var)\phi$ & if &
  $\ga,\assign[a\mapsto\assign(\var)],\state\modelsSL \phi$\\[5pt]
  \multicolumn{3}{l}{If, in addition, $\assign$ is complete, then}\\[5pt]
 $\ga,\assign,\state\modelsSL\X\phi$ & if &
  $\ga,\assigntrans[\path_{\leq 1}]{\assign},\path[1]\modelsSL\phi$, where $\pi$ is the only path in
  $\out(\state,\assign)$\\
 $\ga,\assign,\state\modelsSL\phi\until\phi'$ & if & there is $i\geq 0$
  such that, letting $\path$ be the only path
  in $\out(\state,\assign)$,\\
 & & $\ga,\assigntrans[\path_{\leq i}]{\assign},\path[i]\modelsSL\phi'$, and for all $0\leq j
  <i$, $\ga,\assigntrans[\path_{\leq j}]{\assign},\path[j]\modelsSL\phi$. 
\end{tabular}

Finally, we define an \emph{\SL-sentence} to be an \SL-formula $\phi$ such
that $\free(\phi)=\emptyset$ and every temporal operator in $\phi$ is
under the scope of a binding for each agent.


\section{Branching-time Strategy Logic}
\label{sec-branching}

We now present a first
extension of Strategy Logic. In \SL, temporal
operators are allowed only when every agent has been assigned a
strategy, which leaves only one possible outcome
. Here we relax this constraint: a temporal
formula can be evaluated on the outcome of a \emph{partial} strategy
assignment. The outcome of such an assignment is a tree that contains
all paths corresponding to all possible completions of the assignment, which is why we
use the path quantification of branching-time temporal logic. We also
add the \emph{unbinding} operator as considered in \eg \cite{DBLP:journals/corr/LaroussinieM13}, making it possible to unbind an agent
from its strategy.
We first show that the logic thus obtained, called \BSL, has the same expressivity as
\SL, by providing linear translations in both directions. 
The unbinding operator is thus just convenient syntactic sugar. 
Then we further extend \BSL by allowing it to refer to actions taken
by agents, and obtain the logic \BSLplus that, we postulate,\bm{prove?} is
strictly more expressive than \SL and \BSL. \BSL has two
advantages: first, the semantics is slightly cleaner  than that of
\SL, as it is defined for all formulas and all assignments;
second, the unbinding operator makes it possible to easily express that we unbind
an agent, at no complexity cost. Finally, because it  can
explicitly refer to
 actions and  consider outcomes of partial assignments,
it is possible in \BSLplus  to express properties
of strategies, such as being memoryless or uniform, as we show in Section~\ref{sec-epistemic}.

\subsection{Syntax}

The syntax of \BSL adds two operators to \SL. First, the \emph{path
  quantifier}, borrowed from classic branching-time temporal logics: $\E\psi$ intuitively reads as ``there
exists an outcome of the currently fixed strategies in which $\psi$
holds''. Second, the \emph{unbinding operator}: $(a,\unb)\phi$ means
``$\phi$ holds after Agent~a has been unbound from her strategy, if any''. We define two variants, one
(\BSL) where
formulas cannot talk about the actions taken by the agents, and one
(\BSLplus) where they can. Also, as for \CTLs, we find it convenient to distinguish between state
and path formulas. Finally,  the set of \BSL-formulas (resp. \BSLplus-formulas) is
the set of state formulas given by the following grammar:
\begin{align*}\mbox{State formulas:\hspace{1cm}}&\phi ::= p \mid  \neg\phi \mid \phi\ou\phi   \mid
\Estrat \phi \mid (a,\var)\phi \mid \unbind \phi \mid \E \psi \\
  \mbox{Path formulas:\hspace{1cm}}&\psi ::= \phi \mid \neg \psi \mid
  \psi \ou \psi \mid \X \psi \mid \psi \until \psi, \end{align*}
where $p\in\AP$ (resp. $p\in\APplus)$, $\var\in\setvar$ and $a\in\Ag$. 

Observe that $\BSL\subset\BSLplus$. In addition to the shorthand defined in Section~\ref{sec-SL}, we also
 define the dual of the path quantifier: $\A\phi\egdef\neg\E\neg\phi$.
 Finally, we write \BSLpluspath (resp. \BSLpath) for the set of \BSLplus (resp. \BSL) path
 formulas.

\subsection{Semantics}

State formulas are evaluated in a state of (the unfolding of) a CGS,
and path formulas in paths. Since \BSL is a syntactical fragment of \BSLplus, it is enough to
define the latter's semantics. 

Let  $\phi\in\BSLplus$ be a state
formula (resp. let $\psi\in\BSLpluspath$ be a path
formula), and let $\ga=(\sstates,\delta,q_\init,\gval)$ be
a CGS. Let $\state\in\ga$ be a state, $\path\in\PathsInf$ a path, and let $\assign$ be
an  assignment for $\phi$
(resp. for $\psi$). The semantics of \BSLplus
is defined inductively as follows:

\vspace{5pt}
\begin{tabular}{lcl}
 $\ga,\assign,\state\modelsBSL p$ & if & $p\in \gvalp(\state)$\\
  $\ga,\assign,\state\modelsBSL \neg \phi$ & if & $\ga,\assign,\state
  \not\modelsBSL \phi$  \\
  $\ga,\assign,\state\modelsBSL \phi\ou\phi'$ & if & $\ga,\assign,\state\modelsBSL
 \phi$  or  $\ga,\assign,\state\modelsBSL \phi'$ \\
  $\ga,\assign,\state\modelsBSL \Estrat\phi$ & if & there exists
  $\strat\in\setstrat$ such that $\ga,\assign[\var\mapsto\strat],\state\modelsBSL \phi$\\
$\ga,\assign,\state\modelsBSL (a,\var)\phi$ & if &
$\ga,\assign[a\mapsto\assign(\var)],\state\modelsBSL \phi$\\
$\ga,\assign,\state\modelsBSL \unbind\phi$ & if &
$\ga,\assign[a\mapsto\unb],\state\modelsBSL \phi$\\
 $\ga,\assign,\state\modelsBSL \E\psi$  & if &  there exists
  $\path\in\out(\state,\assign)$ such that
  $\ga,\assign,\path\modelsBSL\psi$\\[5pt]

  $\ga,\assign,\path\modelsBSL \phi$ & if &
  $\ga,\assign,\path[0]\modelsBSL \phi$\\ 
  $\ga,\assign,\path\modelsBSL \neg \psi$  & if & $\ga,\assign,\path \not\modelsBSL \psi$\\
  $\ga,\assign,\path\modelsBSL \psi\ou\psi'$ & if & $\ga,\assign,\path\modelsBSL\psi$ \;or\;
  $\ga,\assign,\path\modelsBSL\psi'$ \\
  $\ga,\assign,\path\modelsBSL\X\psi$ & if &
  $\ga,\assigntrans[\path_{\leq 1}]{\assign},\path_{\geq 1}\modelsBSL\psi$\\
  $\ga,\assign,\path\modelsBSL\psi\until\psi'$ & if &  there is $i\geq 0$ such
  that $\ga,\assigntrans[\path_{\leq i}]{\assign},\path_{\geq
    i}\modelsBSL\psi'$, and\\
& &  for all $0\leq j < i$, $\ga,\assigntrans[\path_{\leq j}]{\assign},\path_{\geq j}\modelsBSL\psi$
\end{tabular}
\vspace{5pt}

The semantics of the unbinding operator comes without surprise:
$\unbind\phi$ holds in an assignment if $\phi$ holds after we have
removed $a$ from the domain of this assignment. For the path
quantifier,  $\E\psi$ holds if there is an outcome of the current
assignment in the current state that verifies $\psi$. 

For a \BSLplus-formula $\phi$, we write $\ga,\assign\modelsBSL\phi$ if $\PathsInf(\state_\init),\assign,\state_\init\modelsBSL\phi$.
Classically, a \emph{\BSLplus-sentence} is a \BSLplus-formula without free
variables, and similarly for \BSL-sentences. For a \BSLplus-sentence $\phi$, we write $\ga\modelsBSL\phi$
if  $\ga,\assign\modelsBSL\phi$ for any assignment $\assign$.


\subsection{Expressivity of \BSL}
\label{sec-comp-SL}



We establish that \BSL and \SL have the same expressivity, and
postulate that\bm{prove?}
\BSLplus is strictly more expressive than both logics.
First, given two logics $\lang$ and $\lang'$ whose sentences are
evaluated on CGS's, we say that $\lang'$  \emph{subsumes}
 $\lang$, written $\lang\subsumed\lang'$, if for every
$\lang$-sentence $\phi$ there is an $\lang'$-sentence $\phi'$ such that,
for every CGS $\ga$, $\ga\models\phi$ if, and only if,
$\ga\models\phi'$. We say that $\lang$ and $\lang'$ are
\emph{equiexpressive} if $\lang\subsumed\lang'$ and $\lang'\subsumed\lang$.
We say that $\lang'$ \emph{strictly subsumes} $\lang$, written $\lang\strsubsumed\lang'$, if
$\lang\subsumed\lang'$ and $\lang'\not\subsumed\lang$.

We start with the easy direction, showing that \BSL subsumes
 \SL.

\begin{definition}
\label{def-tr-SL-BSL}
The translation $\tr:\SL\to\BSL$ is defined by induction as follows:

\begin{tabular}{lclclclclcl}
  $\tr(p)$ & = & $p$ & &
  $\tr(\neg\phi)$ & = & $\neg\tr(\phi)$ & &
  $\tr(\phi \ou \phi')$ & = & $\tr(\phi) \ou \tr(\phi')$ \\
  $\tr(\X\phi)$ & = & $\E\X\tr(\phi)$ & & & & & & 
  $\tr(\phi\until\phi')$ & = & $\E\tr(\phi)\until\tr(\phi')$ \\
  $\tr(\Estrat\phi)$ & = & $\Estrat\tr(\phi)$ & & & & & & 
  $\tr((a,\var)\phi)$ & = & $(a,\var)\tr(\phi)$
\end{tabular}
\end{definition}

The following proposition easily follows from the fact that a complete
assignment defines a unique path from any state.
\begin{proposition}
\label{prop-SL-BSL}
For every CGS $\ga$, for every   formula
$\phi\in\SL$,  assignment
$\assign$ for $\phi$  state $\state\in\ga$ such that
$\ga,\assign,\state\modelsSL\phi$ is defined, it holds that
$\ga,\assign,\state\modelsSL\phi$ if, and only if, $\ga,\assign,\state\modelsBSL\tr(\phi)$.
\end{proposition} 

\begin{proofsketch}
  We only treat the case of the ``next'' operator, the one for
  ``until'' is similar and all the others are trivial.
Assume that $\ga,\assign,\state\modelsSL\X\phi$ is defined. This means
that $\assign$ is a complete assignment, hence $\out(\state,\assign)$
is a singleton, and the result follows from the semantics of \SL and \BSL.
\end{proofsketch}

We now show that \SL also subsumes \BSL. Indeed, the
path quantifier can be simulated by a series of
existential strategy
 quantifications and corresponding bindings for the agents whose
strategies are undefined in the current assignment. Concerning the
unbinding operator, the idea is to remember, along the translation,
which agents have been unbound, and use this information to correctly
translate path quantifiers, as described above. Formally, we define a translation from
\BSL to \SL, parameterized by the set of agents who are  ``currently''
not bound to a strategy.

\begin{definition}
\label{def-tr-BSL-SL}
Let $\bounda\subseteq\Ag$. The translations $\trb_\bounda:\BSL\to\SL$
and $\trbpath_\bounda:\BSLpath\to\SL$ are
defined by mutual induction as follows:

\vspace{5pt}
\begin{tabular}{lclclcl}
  $\trb_\bounda(p)$ & = & $p$ & \hspace{2cm} & $\trbpath_\bounda(\phi)$ & = & $\trb_\bounda(\phi)$\\
 $\trb_\bounda(\neg\phi)$ & = & $\neg\trb_\bounda(\phi)$ & &  $\trbpath_\bounda(\neg\psi)$ & = & $\neg\trbpath_\bounda(\psi)$\\
 $\trb_\bounda(\phi \ou \phi')$ & = & $\trb_\bounda(\phi) \ou
 \trb_\bounda(\phi')$  & &  $\trbpath_\bounda(\psi \ou \psi')$ & = & $\trbpath_\bounda(\psi) \ou
 \trbpath_\bounda(\psi')$ \\
 $\trb_\bounda(\Estrat\phi)$ & = & $\Estrat\trb_\bounda(\phi)$ & &  $\trbpath_\bounda(\X\psi)$ & = & $\X\trbpath_\bounda(\psi)$\\
  $\trb_\bounda((a,\var)\phi)$ & = & $(a,\var)
  \trb_{\bounda\setminus\{a\}}(\phi)$ & &  $\trbpath_\bounda(\psi\until\psi')$ & = & $\trbpath_\bounda(\psi)\until\trbpath_\bounda(\psi')$\\
  $\trb_\bounda(\unbind\phi)$ & = & $ \trb_{\bounda\union\{a\}}(\phi)$
  & & & & \\
  $\trb_\bounda(\E\psi)$ & = &
  \multicolumn{5}{l}{$\Estrat[\var_1]\ldots\Estrat[\var_k](a_{i_1},\var_1)\ldots
  (a_{i_k},\var_k)\trbpath_\bounda(\psi)$,}\\
& & \multicolumn{5}{l}{where
  $\var_1,\ldots,\var_k$ are fresh variables and
  $\{a_{i_1},\ldots,a_{i_k}\}=\bounda$.} 
\end{tabular}
\end{definition}

First, observe that if $p$ is a \BSL-formula, then
  it is in $\AP$ and not in $\APact$, so that $p$ is indeed an \SL
  formula. 
Before establishing the
correctness of the translation, we need the following lemma. It
essentially says that the evaluation of a formula $\trb_\bounda(\phi)$
in an assignment $\assign$ is independent of how $\assign$ is defined
on $\bounda$: for an agent $a\in\bounda$, whether $\assign$ is defined
on $a$ or not, and in the former case how it is defined, is of no
consequence as the translation $\trb_\bounda$ remembers that $a$ is
not supposed to be bound to a strategy. 



\newcounter{lem-extend-assign}
\setcounter{lem-extend-assign}{\value{lemma}}
\begin{lemma}
\label{lem-extend-assign}
Let $\ga$ be a CGS, 
 $\state\in\ga$ a state,   $\phi\in\BSL$  a state formula and
 $\assign$ an assignment  for $\phi$.
For all $A\subseteq \Ag$, $\{a_{i_1},\ldots,a_{i_k}\}\subseteq A$ and
for all
$\strat_{1},\ldots,\strat_{k}\in\setstrats$, letting
$\assign_1=\assign[a_{i_1}\mapsto \strat_1,\ldots,a_{i_k}\mapsto
\strat_k]$ and $\assign_2=\assign[a_{i_1}\mapsto \unb,\ldots,a_{i_k}\mapsto
\unb]$, it holds that:
\begin{itemize}
\item[{\bf P1:}] $\ga,\assign,\state\modelsSL\trb_{A}(\phi)$ if, and only if,
$\ga,\assign_1,\state\modelsSL\trb_{A}(\phi)$, and
\item[{\bf P2:}] $\ga,\assign,\state\modelsSL\trb_{A}(\phi)$ if, and only if,
$\ga,\assign_2,\state\modelsSL\trb_{A}(\phi)$.
\end{itemize}
\end{lemma}



\newcounter{prop-BSL-SL}
\setcounter{prop-BSL-SL}{\value{proposition}}
\begin{proposition}
\label{prop-BSL-SL}
Let $\ga$ be a CGS. For every  state formula $\phi\in\BSL$,  assignment $\assign$
for $\phi$ and state 
$\state\in\ga$,
 it holds that
$\ga,\assign,\state\modelsBSL\phi$ if, and only if,
$\ga,\assign,\state\modelsSL\trb_{\Ag\setminus \dom(\assign)}(\phi)$.
\end{proposition}

We can now prove that \SL and \BSL have the same expressivity on the
level of sentences.

\begin{theorem}
\label{theo-SL-BSL}
\SL and \BSL are equiexpressive, with linear translations in both directions.
\end{theorem}

\begin{proof}
We first prove that $\SL\subsumed\BSL$.  Let $\phi$ be an \SL-sentence. Clearly, $\tr(\phi)$ is a
  \BSL-sentence. Let $\ga$ be a CGS with initial state $\state_\init$, and let $\assign$ be any
  assignment. By definition, $\ga\modelsSL\phi$ iff $\ga,\assign,\state_\init\modelsSL\phi$.
By Proposition~\ref{prop-SL-BSL},
$\ga,\assign,\state_\init\modelsSL\phi$ iff
$\ga,\assign,\state_\init\modelsBSL\tr(\phi)$, and
by definition, the latter is equivalent to $\ga\modelsBSL\tr(\phi)$.

Now, to prove that $\BSL\subsumed\SL$, let $\phi$ be a \BSL-sentence,
and let $\phi'=\trb_\Ag(\phi)$. Observe that $\phi'$ is an
\SL-sentence: indeed, every temporal operator in $\phi$ is under the
scope of some path quantifier, and by definition of $\trb_\Ag$,
every temporal operator in $\phi'$ is thus under the scope of a
binding for each agent. Now, let $\ga$ be a CGS and
 $\assign$ an assignment such that $\Ag\setminus\dom(\assign)=\Ag$.
By definition,
$\ga\modelsBSL\phi$ iff
$\ga,\assign,\state_\init\modelsBSL\phi$ (recall
that since $\phi$ is a sentence, the choice of $\assign$ does not matter for the evaluation of $\phi$). By
Proposition~\ref{prop-BSL-SL}, the latter is equivalent to
$\ga,\assign,\state_\init\modelsSL\phi'$, which by definition is equivalent
to $\ga\modelsSL\phi'$.

Concerning the size of the translations, the one of
Definition~\ref{def-tr-SL-BSL} is clearly linear, and the one in
Definition~\ref{def-tr-BSL-SL} is in $O(2|\Ag||\phi|)$, where $|\Ag|$
is the number of agents and $|\phi|$ the number of symbols in
$\phi$. The translation is thus linear in the size of the formula.
\end{proof}

We can therefore transfer to \BSL the following results known about \SL \cite{DBLP:journals/tocl/MogaveroMPV14}:

\begin{corollary}
\label{cor-model-check}
The model-checking problem for \BSL is nonelementary decidable.
\end{corollary}

\begin{corollary}
\label{cor-model-check}
The  satisfiability problem for \BSL is $\Sigma^1_1$-hard.
\end{corollary}

On the other hand, because \BSLplus can express properties about the
actions taken by agents, it should clearly be strictly more
expressive than \BSL and thus also \SL, but we have not yet proved this. 

\begin{conjecture}
  \BSLplus strictly subsumes \BSL and \SL.
\end{conjecture}


\section{Epistemic Strategy Logic}
\label{sec-epistemic}

In this section, we further extend the framework to account for imperfect
information. For the logic to be expressive enough to express
uniformity of strategies, we need to talk about actions played by the
agents, and we therefore allow the use of atomic propositions in $\APact$.

\subsection{Syntax}
 
We add distributed knowledge operators to the language, one for each
group of agents. The syntax of $\ESL$ is therefore described by the following grammar:
\begin{align*}\mbox{State formulas:\hspace{1cm}}&\phi ::= p \mid  \neg\phi \mid \phi\ou\phi  \mid \E \phi \mid
\Estrat \phi \mid (a,x)\phi \mid \unbind \phi \mid \DA\phi \\
  \mbox{Path formulas:\hspace{1cm}}&\psi ::= \phi \mid \neg \psi \mid
  \psi \ou \psi \mid \X \psi \mid \psi \until \psi, \end{align*}
where $p\in\APplus$, $x\in\setvar$, $a\in\Ag$ and $A\subseteq\Ag$.

We define, for each $a\in\Ag$, $\Ka\phi\egdef\DA[\{a\}]\phi$, and as for \BSL and \BSLplus, we write \ESLpath  for the set of \ESL-path
 formulas.

\subsection{Semantics}

To represent the agents' imperfect information about the current
situation in the game, we add  binary
indistinguishability relations in CGSs. Most works consider
equivalence relations on states, which are extended to initial paths
according to how much memory agents are supposed to have. Because in
this work we do not want to make any such assumptions, we 
adopt a more general approach and directly take equivalence
relations on initial paths. 

We call \emph{imperfect information concurrent game structure}
(ICGS) a tuple $\iga=(\ga,\{\rela\}_{a\in\Ag})$, where $\ga$ is a CGS
and for each $a\in\Ag$, $\rela\;\subseteq (2^{\APplus})^*\times (2^{\APplus})^*$
is an indistinguishability equivalence relation for Agent~$a$. 
For $A\subseteq \Ag$, we let $\relA\egdef\inter_{a\in A}\rela$:
it is the \emph{distributed knowledge relation} of agents
in $A$.
Given two initial
paths $\fpath=\state_0 \acpr_1\state_1\ldots \acpr_n\state_n$ and
$\fpath'=\state'_0\acpr_1\state'_1\ldots \acpr_m\state'_m$ and a set
of agents $A\subseteq\Ag$, we shall write
$\fpath\relA\fpath'$ whenever $\gvalp(\state_0)\ldots\gvalp(\state_n)\relA\gvalp(\state'_0)\ldots\gvalp(\state'_m)$, \ie
when the sequences of extended valuations along the plays are related by
$\relA$.
As usual in epistemic logic, the intended meaning of
$\fpath\rela\fpath'$ is that in initial path $\fpath$, Agent~$a$
considers it possible that $\fpath'$ is the actual  initial path. 

Because agents may infer knowledge from what they recall of the past
of an initial path, we cannot evaluate state formulas merely in states
of the game as we do for \BSLplus, but we evaluate them in initial
paths instead. Also, in order not to forget the past when we consider
outcomes of an assignment, we define for every initial path $\fpath$
and assignment $\assign$, 
$\out(\fpath,\assign)\egdef\{\fpath\cdot\path\mid \path\in\out(\last(\fpath),\assign)\}$.

Let  $\phi\in\ESL$ be a state
formula (resp. let $\psi\in\ESLpath$ be a path
formula), and let $\ga=(\sstates,\delta,q_\init,\gval)$ be
a CGS.  Let $\assign$ be
an  assignment for $\phi$
(resp. for $\psi$), let $\fpath\in\PathsFin$ be an initial path, 
$\path\in\PathsInf$  a path, and $i\geq 0$. The semantics of \ESL
is defined inductively as follows:

\vspace{5pt}
\begin{tabular}{lcl}
 $\iga,\assign,\fpath\modelsESL p$ & if & $p\in \gvalp(\last(\fpath))$\\
  $\iga,\assign,\fpath\modelsESL \neg \phi$ & if & $\iga,\assign,\fpath
  \not\modelsESL \phi$  \\
  $\iga,\assign,\fpath\modelsESL \phi\ou\phi'$ & if & $\iga,\assign,\fpath\modelsESL
 \phi$  or  $\iga,\assign,\fpath\modelsESL \phi'$ \\
  $\iga,\assign,\fpath\modelsESL \Estrat\phi$ & if & there exists
  $\strat\in\setstrat$ such that $\iga,\assign[\var\mapsto\strat],\fpath\modelsESL \phi$\\
$\iga,\assign,\fpath\modelsESL (a,\var)\phi$ & if &
$\iga,\assign[a\mapsto\assign(\var)],\fpath\modelsESL \phi$\\
$\iga,\assign,\fpath\modelsESL \unbind\phi$ & if &
$\iga,\assign[a\mapsto\unb],\fpath\modelsESL \phi$\\
 $\iga,\assign,\fpath\modelsESL \E\psi$  & if &  there exists
  $\path\in\out(\fpath,\assign)$ such that
  $\iga,\assign,\path,|\fpath|\modelsESL\psi$\\
$\iga,\assign,\fpath\modelsESL \DA\phi$ & if & for every initial path 
$\fpath'\in\PathsFin$  such that $\fpath\relA\fpath'$,  $\iga,\assign,\fpath'\modelsESL \phi$
\\[5pt]
  $\iga,\assign,\path,i\modelsESL \phi$ & if &
  $\iga,\assign,\path_{\leq i}\modelsESL \phi$\\ 
  $\iga,\assign,\path,i\modelsESL \neg \psi$  & if & $\iga,\assign,\path,i \not\modelsESL \psi$\\
  $\iga,\assign,\path,i\modelsESL \psi\ou\psi'$ & if & $\iga,\assign,\path,i\modelsESL\psi$ \;or\;
  $\iga,\assign,\path,i\modelsESL\psi'$ \\
  $\iga,\assign,\path,i\modelsESL\X\psi$ & if &
  $\iga,\assigntrans[{\path[i,i+1]}]{\assign},\path,i+1\modelsESL\psi$\\
  $\iga,\assign,\path,i\modelsESL\psi\until\psi'$ & if &  there is $j\geq i$ such
  that $\iga,\assigntrans[{\path[i,j]}]{\assign},\path,j\modelsESL\psi'$, and\\
& &  for all $i\leq k < j$, $\iga,\assigntrans[{\path[i,k]}]{\assign},\path,k\modelsESL\psi$
\end{tabular}
\vspace{5pt}

We now give an example of a property that can be expressed in \ESL but
not in \SL, \BSL or \BSLplus. The property we consider is the
uniformity property of strategies, which is central in
the paradigm of imperfect information.

\subsection{Properties of strategies} 
\label{sec-uniform}

A \emph{uniform strategy}, in the context of games with imperfect information, usually means a strategy that 
respects the player's information, \ie a strategy that assigns the
same action in situations that are indistinguishable to the player
\cite{van2001games,DBLP:journals/fuin/JamrogaH04}. In
\SL, temporal formulas being only evaluated in complete assignments,
it is clear that one cannot compare several outcomes of a given
strategy for a player, so that it is hopeless to express such
uniformity properties. In \BSL, one can consider  all the
possible outcomes of a strategy, but one cannot talk about the actions
taken by agents, so that expressing that a strategy assigns the same
action in different situations is not possible either. In \BSLplus,
we can refer to the precise actions taken by the agents, but we have
no way of relating situations that are indistinguishable to an agent. 
However, as we show below, \ESL is expressive enough for this sort of
properties.

We define a notion of uniformity, that we call
\emph{weak uniformity}, and that  asks for a strategy to be uniform on
all its
outcomes from the current situation. 

\begin{definition}
\label{def-uniform}
Let $\iga=(\ga,\{\rela\}_{a\in\Ag})$ be an ICGS, let
$\fpath\in\PathsFin$ be an initial path and $a\in\Ag$ an agent. A  strategy $\strat$ is
\emph{weakly uniform for $a$ in $\fpath$} if, for all initial paths
$\fplay'\in\out(\fpath,[a\mapsto\strat])$ and $\fplay''\in\PathsFin$ such that
$\fplay'\rela\fplay''$, $\strat(\fplay')=\strat(\fplay'')$.
\end{definition}


Now let us define the following  \ESL-formula.
\begin{definition}
  For each $a\in\Ag$, we define the formula
\[\wuniformaux:=\A\G(\bigou\limits_{\act\in\scriptsize\Act}\!\!\!\Ka \A\X
p^a_{\act})
.\]
\end{definition}

To understand the meaning of this formula, first observe that if an
assignment $\assign$ binds an agent $a$ to a strategy $\strat$, \ie
$\assign(a)=\strat$, then for every initial path $\fpath\in\PathsFin$,
there is an action $\act\in\Act$ such that $p^a_\act$ holds in all
continuations of $\fpath$ of the form
$\fpath'=\fpath\cdot\acpr\state$ that follow $\assign$: this action is
$\strat(\fpath)=\acpr(a)$, the action played by Agent~$a$ in initial path
$\fpath$ according to $\strat$. Therefore,
$\iga,\assign,\fpath\models\A\X p^a_{\strat(\fpath)}$. It follows that, when evaluated in an
initial path $\fpath$ and  assignment $[a\mapsto\strat]$, where $\strat$ is a strategy, formula $\wuniform$
says that at every point of every outcome in $\out(\fpath,\assign)$,
there is an action that Agent~$a$ plays in \emph{all $\rela$-related
  nodes}. 
Let us fix an ICGS $\iga=(\ga,\{\rela\}_{a\in\Ag})$.

\begin{proposition}
\label{prop-wuniform}
For every initial path
$\fpath\in\PathsFin$ and agent $a\in\Ag$, a strategy $\strat$ is
weakly uniform for Agent~$a$ in $\fpath$ if, and only if, $\iga,[a\mapsto\strat],\fpath\models\wuniformaux$.
\end{proposition}


However, $\wuniform$ only has the intended meaning in an assignment
that does not bind any other agent: indeed, otherwise we would only
have that the  strategy considered is uniform on the subset of its
outcomes that follow the strategies assigned to the other agents. 
Consider now the following formula:

\begin{definition}
  For each $a\in\Ag$, noting $\{a_1,\ldots,a_k\}=\Ag\setminus\{a\}$, we define the formula
\[\wuniform:=(a_1,\unb)\ldots (a_k,\unb)\wuniformaux
.\]
\end{definition}

The following proposition holds:

\begin{proposition}
For every initial path
$\fpath\in\PathsFin$, assignment $\assign$ and agent $a\in\Ag$, a strategy $\strat$ is
weakly uniform for Agent~$a$ in $\fpath$ if, and only if, $\iga,\assign[a\mapsto\strat],\fpath\models\wuniform$.
\end{proposition}

We now illustrate how various semantics of \ATL with imperfect
information can be expressed in \ESL. We take the example of the
\ATL formula  $\Estrat[A]\F p$, where $A\subseteq\Ag$. Assume that
$A=\{a_1,\ldots,a_k\}$ and $\Ag\setminus A=\{a_{k+1},\ldots,a_n\}$. We
consider three semantics:  the basic
one from \cite{jamroga2003some}, in which
strategies are just required to be uniform, the \emph{de dicto}
semantics, where in addition the players must know that there is a strategy to
achieve their goal, but may ignore what that strategy is, and the
\emph{de re} semantics, in which there must exist a strategy that the
players  know it ensures their goal (see
\cite{DBLP:journals/fuin/JamrogaH04}, Sec. 3.2). With the first
semantics,  $\Estrat[A]\F p$ would be translated
in \ESL
as:
\[\Estrat[x_1]\ldots \Estrat[x_k](a_1,x_1)\ldots(a_k,x_k) (\biget\limits_{1\leq i \leq
  k}\wuniform[a_i] \et \A\F p).\]

For the \emph{de dicto} semantics, one would write instead: 
\[\DA\Estrat[x_1]\ldots \Estrat[x_k](a_1,x_1)\ldots(a_k,x_k) (\biget\limits_{1\leq i \leq
  k}\wuniform[a_i] \et \A\F p),\]

while for the \emph{de re} semantics, one would write:
\[\Estrat[x_1]\ldots \Estrat[x_k] \DA(a_1,x_1)\ldots(a_k,x_k) (\biget\limits_{1\leq i \leq
  k}\wuniform[a_i] \et \A\F p).\]

One may object that the notion of weak uniformity we consider is too
weak compared to the usual one, which is that a strategy should be
equal on all pairs of related initial paths. We argue that it is
enough for a strategy to be uniform on all the initial paths it may be
involved in while evaluating the formula. 

For instance, in the example
above, the objective is $\A\F p$, so that it is enough to ensure that
strategies for the agents are uniform on their outcome: if a
satisfying set of strategies contains one $\strat_i$ that is not defined
uniformly on some initial paths that are outside its outcome, this
$\strat_i$ can easily be turned into a uniform strategy in the usual
sense, it will still satisfy the formula.

Should we consider a more complex objective, in particular involving knowledge, weak
uniformity may not be sufficient though. Consider the \ESL formula
$\Estrat[x](a,x)\A\G\K_a\A\F p$, where  $a\in\Ag$, which means that
Agent~$a$ wants a strategy such that she always knows that $p$ will
eventually be reached. This objective not only considers outcomes of the
strategy  from the current situation, but also outcomes from initial
paths equivalent to the latter outcomes. In this case, we could strengthen
the requirement on Agent~$a$'s strategy by repeating the weak-uniformity requirement
after each knowledge operator. In the example:
\[\Estrat[x](a,x)(\wuniform \et \A\G\K_a (\wuniform \et \A\F p)).\]

Finally, observe that if we introduced an artificial agent $\amem$ associated to the
relation that relates two initial paths if they end up in the same
state, then the formula $\wuniform[\amem]$ would characterize
strategies that are memoryless on their outcomes from the current
initial path, in the sense that their definition only depends on the
last state of each initial path.




\section{Conclusion}
\label{sec-conclusion}

We have enriched \SL with two operators, the path quantifier and the
unbinding operator, which are convenient but do not add expressivity
in the perfect information case; interestingly though, they do not increase complexity
either. In the context of imperfect information however, these
operators together with knowledge operators and the ability to talk
about actions, allowed us to express properties of strategies which
are usually fixed in the semantics of the logics, such as being
uniform, \emph{de re}, \emph{de dicto}, memoryless\ldots This feature
makes our Epistemic Strategy Logic able to deal with a vast class of
agents without having to change the semantics, and thus unifies many
of the previous proposals in the area.

Of course this comes at a price, and the model-checking problem for
this logic is certainly undecidable with perfect-recall relations and
several agents. We believe that the next steps are, first, to see
whether the syntactical fragments studied for \SL with perfect
information, such as One-Goal or Boolean-Goal Strategy Logic, can be
transferred to \BSL and then to \ESL, and see whether they enjoy
better complexity properties. The second natural move would be to look
at structures which are known to work well with multiple agents with
imperfect information: hierarchical knowledge \cite{BozianuDE13,DBLP:conf/popl/PnueliR89}, recurring
common knowledge of the state \cite{DBLP:journals/corr/BerwangerM14}\ldots

\bibliographystyle{eptcs}
\bibliography{biblio}

\begin{thebibliography}{10}
\providecommand{\bibitemdeclare}[2]{}
\providecommand{\surnamestart}{}
\providecommand{\surnameend}{}
\providecommand{\urlprefix}{Available at }
\providecommand{\url}[1]{\texttt{#1}}
\providecommand{\href}[2]{\texttt{#2}}
\providecommand{\urlalt}[2]{\href{#1}{#2}}
\providecommand{\doi}[1]{doi:\urlalt{http://dx.doi.org/#1}{#1}}
\providecommand{\bibinfo}[2]{#2}

\bibitemdeclare{article}{DBLP:journals/jacm/AlurHK02}
\bibitem{DBLP:journals/jacm/AlurHK02}
\bibinfo{author}{Rajeev \surnamestart Alur\surnameend},
  \bibinfo{author}{Thomas~A. \surnamestart Henzinger\surnameend} \&
  \bibinfo{author}{Orna \surnamestart Kupferman\surnameend}
  (\bibinfo{year}{2002}): \emph{\bibinfo{title}{Alternating-time temporal
  logic}}.
\newblock {\sl \bibinfo{journal}{J. {ACM}}}
  \bibinfo{volume}{49}(\bibinfo{number}{5}), pp. \bibinfo{pages}{672--713},
  \doi{10.1145/585265.585270}.
\newblock \urlprefix\url{http://doi.acm.org/10.1145/585265.585270}.

\bibitemdeclare{inproceedings}{DBLP:journals/corr/Belardinelli14}
\bibitem{DBLP:journals/corr/Belardinelli14}
\bibinfo{author}{Francesco \surnamestart Belardinelli\surnameend}
  (\bibinfo{year}{2014}): \emph{\bibinfo{title}{Reasoning about Knowledge and
  Strategies: Epistemic Strategy Logic}}.
\newblock In: {\sl \bibinfo{booktitle}{SR}}, pp. \bibinfo{pages}{27--33}.
\newblock \urlprefix\url{http://dx.doi.org/10.4204/EPTCS.146.4}.

\bibitemdeclare{inproceedings}{belardinelli2015}
\bibitem{belardinelli2015}
\bibinfo{author}{Francesco \surnamestart Belardinelli\surnameend}
  (\bibinfo{year}{2015}): \emph{\bibinfo{title}{A Logic of Knowledge and
  Strategies with Imperfect Information}}.
\newblock In: {\sl \bibinfo{booktitle}{Private communication}}.

\bibitemdeclare{article}{van2001games}
\bibitem{van2001games}
\bibinfo{author}{Johan \surnamestart van Benthem\surnameend}
  (\bibinfo{year}{2001}): \emph{\bibinfo{title}{{Games in Dynamic-Epistemic
  Logic}}}.
\newblock {\sl \bibinfo{journal}{Bulletin of Economic Research}}
  \bibinfo{volume}{53}(\bibinfo{number}{4}), pp. \bibinfo{pages}{219--248},
  \doi{10.1111/1467-8586.00133}.

\bibitemdeclare{book}{van2011logical}
\bibitem{van2011logical}
\bibinfo{author}{Johan \surnamestart van Benthem\surnameend}
  (\bibinfo{year}{2011}): \emph{\bibinfo{title}{Logical dynamics of information
  and interaction}}.
\newblock \bibinfo{publisher}{Cambridge University Press}.

\bibitemdeclare{inproceedings}{DBLP:journals/corr/BerwangerM14}
\bibitem{DBLP:journals/corr/BerwangerM14}
\bibinfo{author}{Dietmar \surnamestart Berwanger\surnameend} \&
  \bibinfo{author}{Anup~Basil \surnamestart Mathew\surnameend}
  (\bibinfo{year}{2014}): \emph{\bibinfo{title}{Games with recurring
  certainty}}.
\newblock In: {\sl \bibinfo{booktitle}{Proceedings of {SR} 2014}}, pp.
  \bibinfo{pages}{91--96}, \doi{10.4204/EPTCS.146.12}.

\bibitemdeclare{inproceedings}{BozianuDE13}
\bibitem{BozianuDE13}
\bibinfo{author}{R.~\surnamestart Bozianu\surnameend},
  \bibinfo{author}{C.~\surnamestart Dima\surnameend} \&
  \bibinfo{author}{C.~\surnamestart Enea\surnameend} (\bibinfo{year}{2013}):
  \emph{\bibinfo{title}{Model Checking an Epistemic mu-calculus with
  Synchronous and Perfect Recall Semantics}}.
\newblock In: {\sl \bibinfo{booktitle}{TARK'2013}}.
\newblock \urlprefix\url{http://arxiv.org/abs/1310.6434}.

\bibitemdeclare{inproceedings}{DBLP:conf/cav/CermakLMM14}
\bibitem{DBLP:conf/cav/CermakLMM14}
\bibinfo{author}{Petr \surnamestart Cerm{\'a}k\surnameend},
  \bibinfo{author}{Alessio \surnamestart Lomuscio\surnameend},
  \bibinfo{author}{Fabio \surnamestart Mogavero\surnameend} \&
  \bibinfo{author}{Aniello \surnamestart Murano\surnameend}
  (\bibinfo{year}{2014}): \emph{\bibinfo{title}{MCMAS-SLK: A Model Checker for
  the Verification of {S}trategy {L}ogic Specifications}}.
\newblock In: {\sl \bibinfo{booktitle}{CAV}}, pp. \bibinfo{pages}{525--532}.
\newblock \urlprefix\url{http://dx.doi.org/10.1007/978-3-319-08867-9_34}.

\bibitemdeclare{article}{DBLP:journals/iandc/ChatterjeeHP10}
\bibitem{DBLP:journals/iandc/ChatterjeeHP10}
\bibinfo{author}{Krishnendu \surnamestart Chatterjee\surnameend},
  \bibinfo{author}{Thomas~A. \surnamestart Henzinger\surnameend} \&
  \bibinfo{author}{Nir \surnamestart Piterman\surnameend}
  (\bibinfo{year}{2010}): \emph{\bibinfo{title}{Strategy logic}}.
\newblock {\sl \bibinfo{journal}{Inf. Comput.}}
  \bibinfo{volume}{208}(\bibinfo{number}{6}), pp. \bibinfo{pages}{677--693},
  \doi{10.1016/j.ic.2009.07.004}.
\newblock \urlprefix\url{http://dx.doi.org/10.1016/j.ic.2009.07.004}.

\bibitemdeclare{book}{fagin1995reasoning}
\bibitem{fagin1995reasoning}
\bibinfo{author}{Ronald \surnamestart Fagin\surnameend},
  \bibinfo{author}{Joseph~Y. \surnamestart Halpern\surnameend},
  \bibinfo{author}{Yoram \surnamestart Moses\surnameend} \&
  \bibinfo{author}{Moshe~Y. \surnamestart Vardi\surnameend}
  (\bibinfo{year}{1995}): \emph{\bibinfo{title}{Reasoning about knowledge}}.
\newblock \bibinfo{volume}{4}, \bibinfo{publisher}{MIT press Cambridge}.

\bibitemdeclare{article}{halpern2004complete}
\bibitem{halpern2004complete}
\bibinfo{author}{Joseph~Y. \surnamestart Halpern\surnameend},
  \bibinfo{author}{Ron \surnamestart van~der Meyden\surnameend} \&
  \bibinfo{author}{Moshe~Y. \surnamestart Vardi\surnameend}
  (\bibinfo{year}{2004}): \emph{\bibinfo{title}{{Complete Axiomatizations for
  Reasoning about Knowledge and Time}}}.
\newblock {\sl \bibinfo{journal}{SIAM J. Comput.}}
  \bibinfo{volume}{33}(\bibinfo{number}{3}), pp. \bibinfo{pages}{674--703}.
\newblock \urlprefix\url{http://dx.doi.org/10.1137/S0097539797320906}.

\bibitemdeclare{article}{van2003cooperation}
\bibitem{van2003cooperation}
\bibinfo{author}{W.~\surnamestart van~der Hoek\surnameend} \&
  \bibinfo{author}{M.~\surnamestart Wooldridge\surnameend}
  (\bibinfo{year}{2003}): \emph{\bibinfo{title}{{Cooperation, knowledge, and
  time: {A}lternating-time {T}emporal {E}pistemic {L}ogic and its
  applications}}}.
\newblock {\sl \bibinfo{journal}{Studia Logica}}
  \bibinfo{volume}{75}(\bibinfo{number}{1}), pp. \bibinfo{pages}{125--157},
  \doi{10.1023/A:1026185103185}.

\bibitemdeclare{inproceedings}{jamroga2006agents}
\bibitem{jamroga2006agents}
\bibinfo{author}{W.~\surnamestart Jamroga\surnameend} \&
  \bibinfo{author}{T.~\surnamestart {\AA}gotnes\surnameend}
  (\bibinfo{year}{2006}): \emph{\bibinfo{title}{What agents can achieve under
  incomplete information}}.
\newblock In: {\sl \bibinfo{booktitle}{Proceedings of the fifth international
  joint conference on Autonomous agents and multiagent systems}},
  \bibinfo{organization}{ACM}, pp. \bibinfo{pages}{232--234}.

\bibitemdeclare{article}{DBLP:journals/fuin/JamrogaH04}
\bibitem{DBLP:journals/fuin/JamrogaH04}
\bibinfo{author}{Wojciech \surnamestart Jamroga\surnameend} \&
  \bibinfo{author}{Wiebe \surnamestart van~der Hoek\surnameend}
  (\bibinfo{year}{2004}): \emph{\bibinfo{title}{Agents that Know How to Play}}.
\newblock {\sl \bibinfo{journal}{Fundam. Inform.}}
  \bibinfo{volume}{63}(\bibinfo{number}{2-3}), pp. \bibinfo{pages}{185--219}.
\newblock
  \urlprefix\url{http://iospress.metapress.com/content/xh738axb47d8rchf/}.

\bibitemdeclare{article}{jamroga2003some}
\bibitem{jamroga2003some}
\bibinfo{author}{Wojtek \surnamestart Jamroga\surnameend}
  (\bibinfo{year}{2003}): \emph{\bibinfo{title}{Some remarks on alternating
  temporal epistemic logic}}.
\newblock {\sl \bibinfo{journal}{Proceedings of Formal Approaches to
  Multi-Agent Systems (FAMAS 2003)}}, pp. \bibinfo{pages}{133--140}.

\bibitemdeclare{inproceedings}{DBLP:journals/corr/LaroussinieM13}
\bibitem{DBLP:journals/corr/LaroussinieM13}
\bibinfo{author}{Fran{\c{c}}ois \surnamestart Laroussinie\surnameend} \&
  \bibinfo{author}{Nicolas \surnamestart Markey\surnameend}
  (\bibinfo{year}{2013}): \emph{\bibinfo{title}{Satisfiability of {ATL} with
  strategy contexts}}.
\newblock In: {\sl \bibinfo{booktitle}{Proceedings Fourth International
  Symposium on Games, Automata, Logics and Formal Verification, GandALF 2013,
  Borca di Cadore, Dolomites, Italy, 29-31th August 2013.}}, pp.
  \bibinfo{pages}{208--223}, \doi{10.4204/EPTCS.119.18}.
\newblock \urlprefix\url{http://dx.doi.org/10.4204/EPTCS.119.18}.

\bibitemdeclare{inproceedings}{DBLP:conf/concur/MogaveroMPV12}
\bibitem{DBLP:conf/concur/MogaveroMPV12}
\bibinfo{author}{Fabio \surnamestart Mogavero\surnameend},
  \bibinfo{author}{Aniello \surnamestart Murano\surnameend},
  \bibinfo{author}{Giuseppe \surnamestart Perelli\surnameend} \&
  \bibinfo{author}{Moshe~Y. \surnamestart Vardi\surnameend}
  (\bibinfo{year}{2012}): \emph{\bibinfo{title}{What Makes Atl* Decidable? {A}
  Decidable Fragment of Strategy Logic}}.
\newblock In: {\sl \bibinfo{booktitle}{{CONCUR} 2012 - Concurrency Theory -
  23rd International Conference, {CONCUR} 2012, Newcastle upon Tyne, UK,
  September 4-7, 2012. Proceedings}}, pp. \bibinfo{pages}{193--208}.
\newblock \urlprefix\url{http://dx.doi.org/10.1007/978-3-642-32940-1_15}.

\bibitemdeclare{article}{DBLP:journals/tocl/MogaveroMPV14}
\bibitem{DBLP:journals/tocl/MogaveroMPV14}
\bibinfo{author}{Fabio \surnamestart Mogavero\surnameend},
  \bibinfo{author}{Aniello \surnamestart Murano\surnameend},
  \bibinfo{author}{Giuseppe \surnamestart Perelli\surnameend} \&
  \bibinfo{author}{Moshe~Y. \surnamestart Vardi\surnameend}
  (\bibinfo{year}{2014}): \emph{\bibinfo{title}{Reasoning About Strategies: On
  the Model-Checking Problem}}.
\newblock {\sl \bibinfo{journal}{{ACM} Trans. Comput. Log.}}
  \bibinfo{volume}{15}(\bibinfo{number}{4}), pp. \bibinfo{pages}{34:1--34:47},
  \doi{10.1145/2631917}.
\newblock \urlprefix\url{http://doi.acm.org/10.1145/2631917}.

\bibitemdeclare{inproceedings}{DBLP:conf/fsttcs/MogaveroMV10}
\bibitem{DBLP:conf/fsttcs/MogaveroMV10}
\bibinfo{author}{Fabio \surnamestart Mogavero\surnameend},
  \bibinfo{author}{Aniello \surnamestart Murano\surnameend} \&
  \bibinfo{author}{Moshe~Y. \surnamestart Vardi\surnameend}
  (\bibinfo{year}{2010}): \emph{\bibinfo{title}{Reasoning About Strategies}}.
\newblock In: {\sl \bibinfo{booktitle}{{IARCS} Annual Conference on Foundations
  of Software Technology and Theoretical Computer Science, {FSTTCS} 2010,
  December 15-18, 2010, Chennai, India}}, pp. \bibinfo{pages}{133--144},
  \doi{10.4230/LIPIcs.FSTTCS.2010.133}.
\newblock \urlprefix\url{http://dx.doi.org/10.4230/LIPIcs.FSTTCS.2010.133}.

\bibitemdeclare{inproceedings}{DBLP:conf/popl/PnueliR89}
\bibitem{DBLP:conf/popl/PnueliR89}
\bibinfo{author}{A.~\surnamestart Pnueli\surnameend} \&
  \bibinfo{author}{R.~\surnamestart Rosner\surnameend} (\bibinfo{year}{1989}):
  \emph{\bibinfo{title}{On the Synthesis of a Reactive Module}}.
\newblock In: {\sl \bibinfo{booktitle}{POPL'89}}, pp.
  \bibinfo{pages}{179--190}, \doi{10.1145/75277.75293}.
\newblock \urlprefix\url{http://doi.acm.org/10.1145/75277.75293}.

\bibitemdeclare{article}{reif84}
\bibitem{reif84}
\bibinfo{author}{John~H. \surnamestart Reif\surnameend} (\bibinfo{year}{1984}):
  \emph{\bibinfo{title}{{The complexity of two-player games of incomplete
  information}}}.
\newblock {\sl \bibinfo{journal}{Journal of computer and system sciences}}
  \bibinfo{volume}{29}(\bibinfo{number}{2}), pp. \bibinfo{pages}{274--301},
  \doi{10.1016/0022-0000(84)90034-5}.

\bibitemdeclare{article}{schobbens2004alternating}
\bibitem{schobbens2004alternating}
\bibinfo{author}{Pierre-Yves \surnamestart Schobbens\surnameend}
  (\bibinfo{year}{2004}): \emph{\bibinfo{title}{Alternating-time logic with
  imperfect recall}}.
\newblock {\sl \bibinfo{journal}{Electronic Notes in Theoretical Computer
  Science}} \bibinfo{volume}{85}(\bibinfo{number}{2}), pp.
  \bibinfo{pages}{82--93}.

\end{thebibliography}



\end{document}